\documentclass[twocolumn, journal]{IEEEtran}

\usepackage[latin1]{inputenc}
\usepackage{graphicx}
\usepackage{amssymb,amsmath,amsfonts,amsthm}
\usepackage{algorithm}
\usepackage{cite}
\usepackage{float}
\usepackage{graphics}
\usepackage{subfigure}
\usepackage{xspace}
\usepackage[usenames,dvipsnames]{color}
\usepackage{amssymb, epsfig, hyperref}
\usepackage[noend]{algorithmic}
\usepackage{bbm}
\usepackage{subeqnarray}

\newtheorem{theorem}{Theorem}

\newtheorem{proposition}{Proposition}

\begin{document}

\title{Fairness for Non-Orthogonal Multiple Access \\ in 5G Systems}

\vspace{-0.6cm}
\author{Stelios Timotheou,~\IEEEmembership{Member,~IEEE}, and Ioannis Krikidis,~\IEEEmembership{Senior Member,~IEEE}
\vspace{-0.4cm}
\thanks{S. Timotheou and I. Krikidis are with the KIOS Research Center for Intelligent Systems and Networks, University of Cyprus, Nicosia 1678 (E-mail: \{timotheou.stelios, krikidis\}@ucy.ac.cy)}
\thanks{Part of this work was supported by the Research Promotion Foundation, Cyprus under the project KOYLTOYRA/BP-NE/0613/04 ``Full-Duplex Radio: Modeling, Analysis and Design (FD-RD)''.}}

\maketitle

\begin{abstract}
In non-orthogonal multiple access (NOMA) downlink, multiple data flows are superimposed in the power domain and user decoding is based on successive interference cancellation. NOMA's performance highly depends on the power split among the data flows and the associated power allocation (PA) problem. In this letter, we study NOMA from a fairness standpoint and we investigate PA techniques that ensure fairness for the downlink users under i) instantaneous channel state information (CSI) at the transmitter, and ii) average CSI.  Although the formulated problems are non-convex, we have developed low-complexity polynomial algorithms that yield the optimal solution in both cases considered. 
\end{abstract}

\vspace{-0.1cm}

%
%
%
%

\vspace{-0.3cm}
\section{Introduction}

\IEEEPARstart{F}{uture} 5G communication systems need to support unprecedented requirements for the wireless access connection, targeting cell throughput capacities of $1000\times$  current 4G technology and roundtrip latency of about $1$ msec. Towards this direction, three major 5G technologies named ultra-densification, millimeter wave, and massive multiple-input multiple-output, have attracted considerable  attention in both industry and academia \cite{AND}. In addition to these technological implications, physical layer issues such as transmission waveforms and multiple-access (MA) schemes should be reconsidered.  A promising downlink MA scheme is the non-orthogonal multiple access (NOMA) which achieves high spectral efficiencies by combining superposition coding at the transmitter with successive interference cancellation (SIC) at the receivers \cite{SAI, KIM}. 

In contrast to the conventional orthogonal MA schemes (e.g. time-division multiple access (TDMA), etc.), NOMA simultaneously serves multiple users in the same degrees of freedom by splitting them in the power domain. Since NOMA is based on the SIC order, the served users achieve unequal rates and this could be critical for scenarios with strict fairness constraints. The work in \cite{DIN} analyzes the performance of the NOMA scheme in terms of outage probability and achievable sum-rate for a fixed power allocation without discussing potential fairness issues. In order to enable fairness, \cite{DIN2} introduces a cooperative phase after the NOMA downlink, where strong users relay data for the weak users but with the cost of extra channel resources (i.e., dedicated time slots are used for cooperation). On the other hand, fairness can be supported through appropriate power allocation (PA) of the  superimposed transmitted data flows.  

The main novelty of this work is the investigation of the impact of PA on the fairness performance of the NOMA scheme. In particular, we study  the PA problem in NOMA from a fairness standpoint under two main system assumptions: i) when  user's data rates are adopted to the channel conditions (perfect channel state information (CSI)), and ii) when users have fixed targeted data rates under an average CSI. Although the resulting optimization problems are both non-convex, we develop low-complexity bisection-based iterative algorithms that provably yield globally optimal solutions. We further illustrate that each iteration subproblem can be optimally solved in closed and semi-closed form for the two cases, with no specialized optimization software. The results show that the NOMA scheme outperforms conventional MA approaches by significantly improving the performance of the worst user. 

\vspace{-0.1cm}
\section{System Model}
\label{sec:model}

We assume a single-cell downlink topology consisting of one base station, $B$, and $N$ users, $U_i$, with $i\in\mathcal{N}=\{1,\ldots,N\}$; all terminals are equipped with a single antenna. $B$ has always data to transmit for each user (saturated scenario) and its total available transmitted power is equal to $P$. All wireless links exhibit independent and identically distributed (i.i.d.) block Rayleigh fading and additive white Gaussian noise (AWGN). This means that the fading coefficients $h_i$ (for the $B\rightarrow U_i$ link) remain constant during one slot, but change independently from one slot to another according to a complex Gaussian distribution with zero mean and variance $\sigma_h^2$; the variance captures path-loss and shadowing effects. Without loss of generality,  the channels are sorted as $0<|h_1|^2 \leq |h_2|^2\leq \hdots |h_N|^2$ i.e., the $i$-th user always holds the $i$-th  weakest instantaneous channel. The AWGN is assumed to be normalized with zero mean and variance $\sigma_n^2$.

\vspace{-0.3cm}
\subsection{NOMA scheme}

The NOMA scheme allows $B$ to simultaneously serve all users by using the entire bandwidth to transmit data via a superposition coding technique at the transmitter side and SIC techniques at the users \cite[Eq. (6.25)]{tse}; in this case, user multiplexing is performed in the power domain.  According to the NOMA principles, $B$ transmits a linear superposition of $N$ data flows by allocating a fraction $\beta_i$ of the total power to the $i$-th data flow. Each receiver employs a SIC technique and is able to perfectly decode the signals of the weakest users i.e., the $i$-th user can decode the signal for the $m$-th user with $m\leq i$,  removing completely inter-user interference. In this way, the $i$-th user can remove interference from the weakest users and its achievable rate associated with the decoding of the $m$-th data flow is given by \cite{DIN}
\begin{align}
R_{i,m}^{\text{NOMA}}(\boldsymbol \beta)=\log\left(1+\frac{\beta_m P |h_i|^2}{P|h_i|^2\sum_{k=m+1}^N\beta_k+\sigma_n^2} \right),
\end{align} 
where $R_{m,m}\leq R_{i,m}$, $m\le i$ due to the ordering of the channel coefficients.

\subsection{Outage probability performance for NOMA}

In case $B$ does not have any instantaneous channel feedback, it transmits with a target spectral efficiency $r_0$ bits per channel use (BPCU) for each data flow and an appropriate performance metric is the outage probability.  An outage event occurs at the $i$-th user when is not able to decode  its own data flow or the data flows of the weakest users $m<i$. By using high order statistics and following similar steps to \cite{DIN}, the outage probability for the $i$-th user is 
\begin{subequations}
\begin{align}
&\mathcal{P}_i^{\textrm{NOMA}}(\boldsymbol \beta)=1-\mathbb{P}\bigg\{\bigcap_{m=1,\ldots,i}R_{i,m}^{\text{NOMA}}(\boldsymbol \beta)\geq r_0  \bigg\} \nonumber \\
&=\mathbb{P}\bigg\{\bigcap_{m=1,\ldots,i} |h_i|^2\geq \zeta_m  \bigg\} = \mathbb{P}\bigg\{|h_i|^2\geq \hat{\zeta}_i \bigg\}  \label{eq:outageDefinition1} \\
&=\Delta_i \int_0^{\hat{\zeta}_i}\bigg(1-\exp(-\lambda x) \bigg)^{i-1}\exp\big(-\lambda (N-i+1)x\big)dx \nonumber \\
&=\! \sum_{k=0}^{i-1}\!\gamma_{i,k}\bigg(\!1\!-\!\exp(-\delta_{i,k}\hat{\zeta}_i)\! \bigg), \label{out2}
\end{align}
\end{subequations}
where $\gamma_{i,k}=\frac{\Delta_i \binom{i-1}{k}(-1)^k}{\delta_{i,k}}$, $\Delta_i = \frac{N!}{(i-1)!(N-i)!}$, $\delta_{i,k}=\lambda(N-i+1+k)$, $\hat{\zeta}_i=\max\{\zeta_1,\zeta_2,\ldots,\zeta_i \}$, $\lambda = 1/\sigma_h^2$, $\zeta_i=\frac{\sigma_n^2\hat{r}_0}{P(\beta_i-\hat{r}_0 \sum_{l=i+1}^N\beta_l)}$, $\hat{r}_0=2^{r_0}-1$, $\beta_i>\hat{r}_0 \sum_{l=i+1}^N \beta_l$ (see Eq. (8) in \cite{DIN}), and \eqref{out2} follows from the binomial theorem. It is worth noting that here, we give a closed form expression for i.i.d. Rayleigh fading, while the analysis in \cite{DIN} refers to a different channel probability distribution. 

\section{Fairness for NOMA systems}
The NOMA scheme enables a more flexible management of the users' achievable rates and is an efficient way to enhance user fairness.  In this section, we consider two fundamental fairness criteria that refer to NOMA systems with instantaneous and average CSI, respectively.

\subsection{Max-Min fairness with instantaneous CSI}\label{sec:perfCSI}
If a continuous channel feedback is available at the transmitter side, users' rates can be allocated according to their instantaneous channel conditions. In this case, a suitable criterion is the max-min fairness that maximizes the minimum achievable user rate and is formulated as follows
\begin{subequations}
\label{eq:maxminNOMA}
\begin{align}
&\max_{\boldsymbol \beta} {\displaystyle \min_{i\in\mathcal{N}}}\; R_{i,i}^{\textrm{NOMA}}(\boldsymbol \beta), \label{eq:maxminNOMAa}\\
&\textrm{s.t.}\;\; \sum_{j=1}^N\beta_j \leq 1,  \label{eq:maxminNOMAb} \\
&0 \leq \beta_j,\;\;\;\textrm{for}\;j\in\mathcal{N}.  \label{eq:maxminNOMAc} 
\end{align}
\end{subequations}
Problem \eqref{eq:maxminNOMA} is not convex and hence hard to solve directly using standard optimization solvers. In this section, we transform the problem into a sequence of linear programs (LPs) and develop a customized low-complexity polynomial algorithm for its optimal solution. Towards this direction, we state proposition \ref{prop:maxminNOMA}. 
\begin{proposition}
\label{prop:maxminNOMA}
Problem \eqref{eq:maxminNOMA} is quasi-concave. 
\end{proposition}

\begin{proof}
A maximization optimization problem is quasi-concave when the objective function is quasi-concave and the constraints are convex. Clearly, the constraints of \eqref{eq:maxminNOMA} are convex since  \eqref{eq:maxminNOMAb} and \eqref{eq:maxminNOMAc} are linear. For the objective function to be quasi-concave, all its sublevel sets must be concave \cite{boyd}, i.e., $\mathcal{S}_t=\{{\displaystyle \min_i} R_{i,i}^{\textrm{NOMA}}(\boldsymbol \beta)\ge t\}$, for $t\in \mathbb{R}$,  which denotes the set of $\boldsymbol \beta$ for which the objective function is larger than $t$. Due to the $\min$ operator, it is true that $\mathcal{S}_t=\{R_{i,i}^{\textrm{NOMA}}(\boldsymbol \beta)\ge t, i\in\mathcal{N}\}$. Nonetheless, set $\mathcal{S}_t$ is concave for $t\in \mathbb{R}$, as constraints $R_{i,i}^{\textrm{NOMA}}(\boldsymbol \beta)\ge t, i\in\mathcal{N}$ can be expressed as
\begin{align}
\beta_i P |h_i|^2\ge (2^t-1)\left(P|h_i|^2\sum_{l=i+1}^N\beta_l+\sigma_n^2\right),~i\in\mathcal{N},\label{eq:biLinear}
\end{align}
which are linear inequalities and hence concave, completing the proof.
\end{proof}

\begin{algorithm}[t]
	\caption{\textbf{: Optimal solution to problem \eqref{eq:maxminNOMA}}} 
	\begin{algorithmic} [1] 
 	\STATE  \textbf{Init.} $t_{LB}=0$,~$t_{UB}=\log\left(1+\frac{P |h_N|^2}{\sigma_n^2}\right)$.
  \WHILE{($t_{UB}- t_{LB}\ge\epsilon$)}
	\STATE Set $t= (t_{UB} + t_{LB})/2$; Solve LP \eqref{eq:seqLP2} to obtain $\boldsymbol\beta^{LP}$. 
	\IF{$\left({\displaystyle \sum_{i\in\mathcal{N}}} \beta_i^{LP} \le 1\right)$}
		\STATE Set $t_{LB}=t$; $\boldsymbol \beta^* = \boldsymbol\beta^{LP}$; $r^*=t$.	
	\ELSE
		\STATE Set $t_{UB}=t$.	
	\ENDIF
	\ENDWHILE
\end{algorithmic}  
\label{alg:bisect}
\end{algorithm}

Let $r^*$ denote the optimal objective function value to quasi-concave problem \eqref{eq:maxminNOMA}. For a specific constant value $t$, if the LP 
\begin{align}
\text{Find}\;\;  \boldsymbol \beta\;\;  \text{subject to constraints \eqref{eq:maxminNOMAb}, \eqref{eq:maxminNOMAc}, \eqref{eq:biLinear}},  \label{eq:seqLP1}
\end{align}
is feasible then $r^*\ge t$, otherwise  $r^*\le t$. Equivalently, one can solve the following LP
 \begin{align}
{\displaystyle \min_{\boldsymbol \beta}} {\displaystyle \sum_{i\in\mathcal{N}}} \beta_i \text{ subject to constraints \eqref{eq:maxminNOMAc} and \eqref{eq:biLinear},}  \label{eq:seqLP2}
\end{align}
and check if the solution satisfies \eqref{eq:maxminNOMAb}. This implies that by appropriately bounding $t$ through a bisection procedure (Algorithm \ref{alg:bisect}), the optimal solution to  \eqref{eq:maxminNOMA}, within a desirable accuracy $\epsilon$, can be obtained by solving a sequence of feasibility LPs of the form \eqref{eq:seqLP2}.

Although \eqref{eq:seqLP2} is LP and can be solved with standard optimization solvers, its solution can be obtained in closed-form with linear computational complexity from a customized algorithm. Towards this direction Proposition \ref{prop:equalityNOMA} is essential.
\begin{proposition}
\label{prop:equalityNOMA}
The optimal solution of \eqref{eq:seqLP2}, satisfies all constraints \eqref{eq:maxminNOMAc} with strict inequality (\emph{inactive constraints}) and \eqref{eq:biLinear} with equality (\emph{active constraints}). 
\end{proposition}

\begin{proof}
Because the problem is convex, the following Karush-Kuhn-Tucker (KKT) conditions are necessary and sufficient for optimality of \eqref{eq:seqLP2}, for $i\in\mathcal{N}$:
\begin{align}
&\lambda_i P |h_i|^2 + \mu_i = \sum_{k<i}\lambda_k (2^t-1)P|h_k|^2 + 1,\label{eq:KKTNOMA1}\\
&\beta_i P |h_i|^2 \ge (2^t-1)\left(P|h_i|^2\sum_{l=i+1}^N\beta_l+\sigma_n^2\right),\label{eq:KKTNOMA2}\\
&\beta_i\ge0,~\lambda_i\ge0, ~\mu_i\ge0,\label{eq:KKTNOMA3}\\
&\lambda_i\left((2^t-1)\left(P|h_i|^2\sum_{l=i+1}^N\beta_l+\sigma_n^2\right) - \beta_i P |h_i|^2\right)=0,\label{eq:KKTNOMA4}\\
&\mu_i\beta_i=0,\label{eq:KKTNOMA5}
\end{align}
where $\lambda_i$ and $\mu_i$, $i\in\mathcal{N}$, are the Lagrange multipliers for constraints \eqref{eq:biLinear} and \eqref{eq:maxminNOMAc} respectively. 
The right hand side (r.h.s.) of \eqref{eq:KKTNOMA2} is strictly positive for all $i\in\mathcal{N}$, as $\sigma_n^2>0$, $t>0$ and $\beta_i\ge0$; hence, the left hand side (l.h.s.) has to be strictly positive which implies that $\beta_i>0$ and $\mu_i=0$ (due to \eqref{eq:KKTNOMA5}), which completes the first part of the proposition. In a similar fashion the r.h.s. of \eqref{eq:KKTNOMA1} is strictly positive, which implies that $\lambda_i>0$, $i\in\mathcal{N}$, as $\mu_i=0$. Since $\lambda_i>0$, condition \eqref{eq:KKTNOMA4} implies that all constraints \eqref{eq:KKTNOMA2} must be enforced with equality which completes the proof. 
\end{proof}
The intuition behind proposition \ref{prop:equalityNOMA} is that all $\beta_i$, $i\in\mathcal{N}$ must be positive for the problem constraints to be satisfied, while all  $R_{i,i}$ must be equal at the optimal solution. Based on proposition \ref{prop:equalityNOMA}, the optimal solution to LP \eqref{eq:seqLP2} can be obtained from Theorem \ref{th:1}.
\begin{theorem}
\label{th:1}
The optimal solution to \eqref{eq:seqLP2} is given by:
\begin{align}
&\beta_i = \frac{2^t-1}{P |h_i|^2 }\left(P|h_i|^2\sum_{l=i+1}^N\beta_l+\sigma_n^2\right), i=N,\ldots,1.\label{eq:optimalNOMA}
\end{align}
\end{theorem}
 
\begin{proof} 
The proof is a direct consequence of the active inequality constraints \eqref{eq:biLinear} from Proposition \ref{prop:equalityNOMA}. The analytical solution emanates from the fact that normalized power $\beta_i$ only depends on the power allocated to the stronger channels $i+1,\ldots,N$. Hence, the optimal objective function value can be computed by allocating power from the strongest to the weakest channel in succession. 
\end{proof} 

Notice that due to the special structure of the problem, the optimal solution can be obtained in closed form (see \eqref{eq:optimalNOMA}), with all users having data rate equal to $t$. In addition, it can be easily observed that the computational complexity of solving \eqref{eq:seqLP2} is $O(N)$, i.e., linear to the number of users. 

\vspace{-0.2cm}
\subsection{Min-Max fairness with average CSI}

Optimizing the performance of communication systems under average CSI information is an important and challenging problem with practical interest. Towards this direction, we propose to optimize the outage probability of all users under knowledge of the distribution and order of channels in NOMA systems. The considered optimization problem is defined as
\begin{subequations}
\label{eq:outageNOMA}
\begin{align}
&\min_{\boldsymbol \beta} {\displaystyle \max_{i\in\mathcal{N}}}\;\; \mathcal{P}_i^{\textrm{NOMA}}(\boldsymbol \beta),  \label{eq:outageNOMAa}\\
&\textrm{subject to constraints \eqref{eq:maxminNOMAb} and \eqref{eq:maxminNOMAc}}, 
\end{align}
\end{subequations}
where $\mathcal{P}_i^{\textrm{NOMA}}(\boldsymbol\beta)$ is given by \eqref{out2}. Naturally, increasing the total power monotonically decreases the optimal outage probability of \eqref{eq:outageNOMA}. Hence, a bisection procedure similar to Algorithm \ref{alg:bisect} can be followed, where each step involves the solution of the following problem
\begin{subequations}
\label{eq:outageNOMA2}
\begin{align}
\min_{\boldsymbol \beta} & {\displaystyle \sum_{i\in\mathcal{N}}} \beta_i,\label{eq:outageNOMA2a}\\
\text{s.t.}\;\; & \mathcal{P}_i^{\textrm{NOMA}}(\boldsymbol \beta)\le t, ~~~\beta_i\ge 0.\label{eq:outageNOMA2c}
\end{align}
\end{subequations}
\vspace{-0.5cm}

\noindent For a particular value $t$, if the solution of \eqref{eq:outageNOMA2} satisfies $\sum_{i\in\mathcal{N}}\beta_i\le 1$ then the optimal solution of \eqref{eq:outageNOMA} is $r^*\le t$, otherwise $r^*> t$. Hence, the optimal solution of \eqref{eq:outageNOMA}, can be obtained by optimally solving a sequence of \eqref{eq:outageNOMA2}. Nonetheless, problem \eqref{eq:outageNOMA2} is non-convex due to the presence of terms $\exp(-\delta_{i,k}\hat{\zeta}_i)$ in \eqref{eq:outageNOMA2a} which are neither convex or concave. To simplify the problem, the next proposition allows the elimination of $\hat{\zeta}_i$. 

\begin{proposition}
\label{prop:zeta}
At the optimal solution of problem \eqref{eq:outageNOMA2}, it is true that $|h_1|^2 = \zeta_1 \le \ldots \le |h_N|^2=\zeta_N$.
\end{proposition}

\begin{proof}
The proposition will be proved by induction. Assume that the optimal solution of \eqref{eq:outageNOMA2} is $\boldsymbol \beta^*$. 
From the definition of NOMA outage probability \eqref{eq:outageDefinition1} for $i=1$, we have that $\mathbb{P}\left\{ |h_1|^2\geq \zeta_1 \right\}$
=$\mathbb{P} \left\{ |h_1|^2 \geq \frac{\sigma_n^2\hat{r}_0}{P(\beta^*_1-\hat{r}_0 \sum_{k=2}^N\beta^*_k)}\right\}$. Since $\beta^*_1$ only affects the outage probability of user 1, we want to select $\beta_1$ as small as possible to minimize the total used power which implies that $\mathbb{P}\left\{|h_1|^2=\zeta_1\right\}$. Assume that $|h_1|^2=\zeta_1,\ldots,|h_k|^2=\zeta_k$; we will prove the result for $k+1$. From \eqref{eq:outageDefinition1} we have that  $\mathbb{P}\{\bigcap_{m=1,\ldots,k+1} |h_{m}|^2\geq \zeta_{m}\}$ =  $\mathbb{P}\left\{|h_{k+1}|^2\geq \zeta_{k+1}\right\}$, as from the above assumption is it true that $|h_{k+1}|^2 \geq \zeta_{m}=|h_{m}|^2$, $m=1,\ldots,k$. Because $\beta_{k+1}$ does not affect the outage probability of users $k+2$ to $N$ it can be selected to minimize the total consumed power which is true when $|h_{k+1}|^2=\zeta_{k+1}$. This completes the proof. 
\end{proof}

\begin{figure*}[t]
  \begin{minipage}{0.31\linewidth}\centering
    \includegraphics[width=\linewidth]{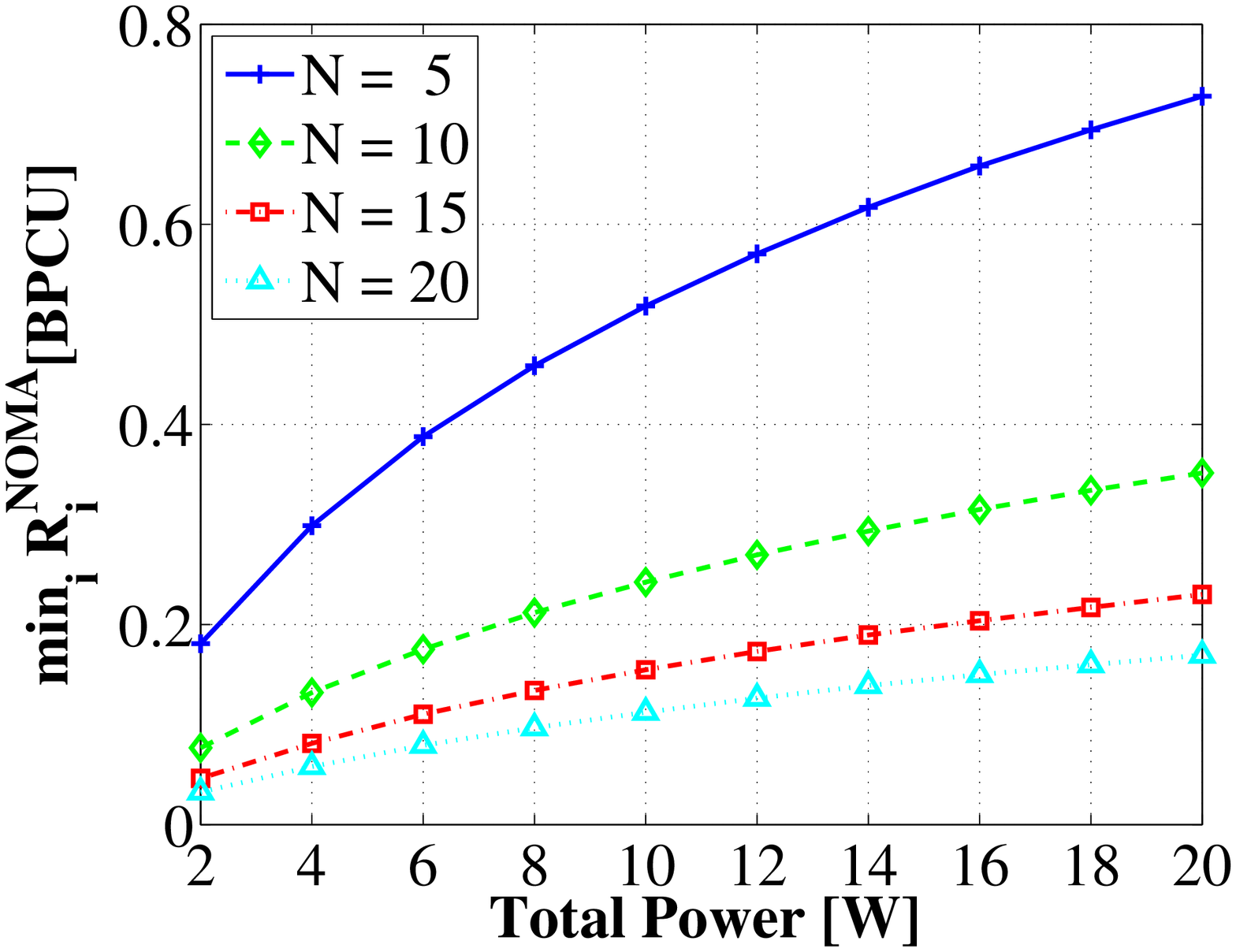}\label{Fig1} 
    \caption{Optimal NOMA fairness rate for different number of users and total power under instantaneous CSI.}
  \end{minipage}
  \begin{minipage}{0.31\linewidth}\centering
    \includegraphics[width=\linewidth]{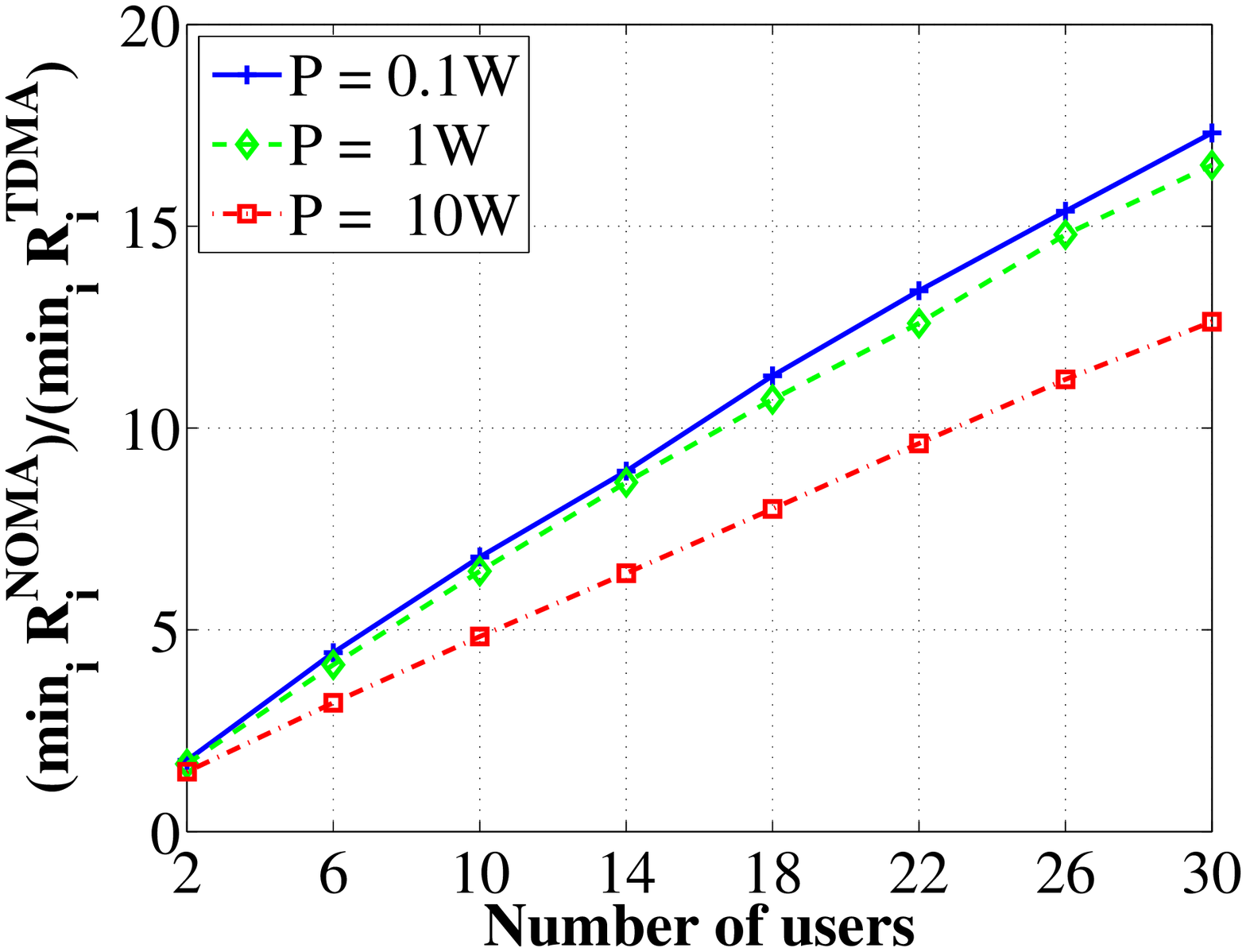}\label{Fig2} 
    \caption{Comparative NOMA and TDMA fairness rate performance under instantaneous CSI for different number of users and total power.}
  \end{minipage}
  \begin{minipage}{0.31\linewidth}\centering
    \includegraphics[width=\linewidth]{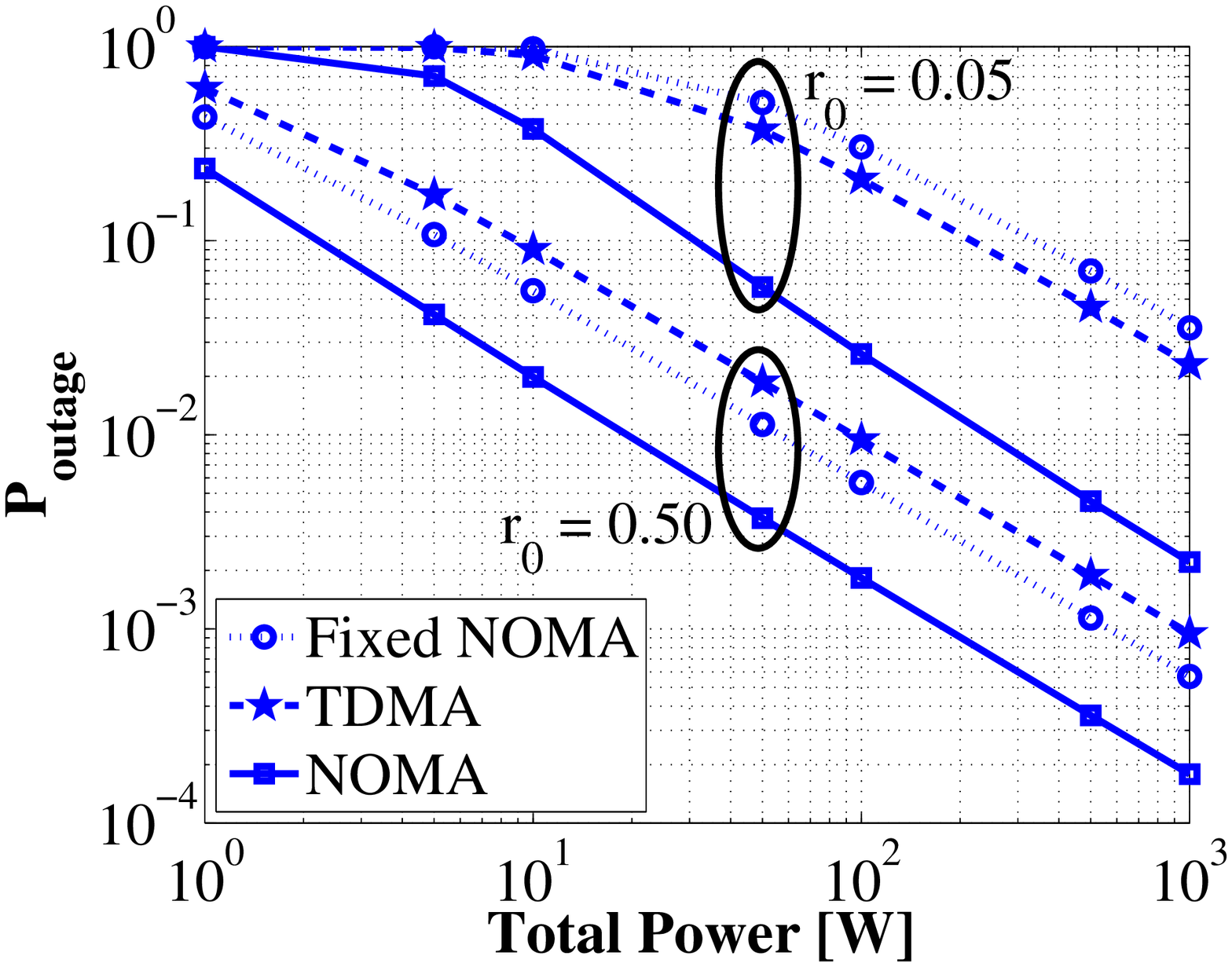}\label{Fig3} 
    \caption{Comparative outage probability performance between NOMA, TDMA and fixed NOMA under average  CSI.}
  \end{minipage}
\end{figure*}

The intuition behind proposition \ref{prop:zeta} is that the target spectral efficiency needs to be satisfied with equality for all users. Proposition \ref{prop:zeta}, implies that at the optimal solution of \eqref{eq:outageNOMA2}, it is true that $\hat{\zeta}_i=\zeta_i$. Next, we show that despite of \eqref{eq:outageNOMA2} being non-convex, its globally optimal solution can be obtained in semi-closed form.

\vspace{-0.2cm}
\begin{theorem}
The optimal solution of \eqref{eq:outageNOMA2} is given by:
\begin{align}
\label{eq:optimalBetaNOMA}
&\beta_i^* = \frac{\sigma_n^2\hat{r}_0}{P\zeta^*_i} + \hat{r}_0 \sum_{k=i+1}^N\beta^*_k, i=N,\ldots,1, 
\end{align}
where $\zeta_i^*$, $i\in\mathcal{N}$ denotes the solution of the 1-D equation
\begin{align}
\label{eq:optimalZetaNOMA}
&\sum_{k=0}^{i-1} \gamma_{i,k}\bigg(1-\exp(-\delta_{i,k}\zeta_i)\bigg) = t,  
\end{align}
\end{theorem}
\begin{proof}
To prove the above theorem, variable transformation $\zeta_i=\frac{\sigma_n^2\hat{r}_0}{P(\beta_i-\hat{r}_0 \sum_{l=i+1}^N\beta_l)}$ is made on problem \eqref{eq:outageNOMA2} yielding
\begin{align}
\min & ~\sum_{i\in\mathcal{N}} \sigma_n^2\hat{r}_0(1+\hat{r}_0)^{i-1}/(P\zeta_i), \label{eq:problemZeta}\\
& \sum_{k=0}^{i-1} \gamma_{i,k}\bigg(1-\exp(-\delta_{i,k}\zeta_i)\bigg)\le t, ~\zeta_i\ge 0, ~i\in\mathcal{N}. \nonumber 
\end{align}

Interestingly, the constraints in problem \eqref{eq:problemZeta}  are uncoupled and the objective is additive; hence, the problem can be decomposed into $N$ one-dimensional problems, where the $i$-th subproblem is of the form
\begin{subequations}
\label{eq:problemZetaSub}
\begin{align}
\min & ~\sigma_n^2\hat{r}_0(1+\hat{r}_0)^{i-1}/(P\zeta_i), \label{eq:problemZetaSuba}\\
& \sum_{k=0}^{i-1} \gamma_{i,k}\bigg(1-\exp(-\delta_{i,k}\zeta_i)\bigg)\le t,\label{eq:problemZetaSubb}\\
& \zeta_i\ge 0. \label{eq:problemZetaSubc}
\end{align}
\end{subequations}
Although, subproblems \eqref{eq:problemZetaSub} are non-convex (due to \eqref{eq:problemZetaSubb}), it can easily be observed that \eqref{eq:problemZetaSuba} is monotonically decreasing in $\zeta_i$, while constraint \eqref{eq:problemZetaSubb} is monotonically increasing in $\zeta_i$. Hence, the optimal solution to \eqref{eq:problemZetaSub} is the largest value of $\zeta_i$ that satisfies \eqref{eq:problemZetaSubb}, which is obtained when the constraint is active. Having obtained the optimal solution of each subproblem $\boldsymbol \zeta^*$, the optimal solution to $\boldsymbol \beta^*$ can be obtained by back-substitution yielding \eqref{eq:optimalBetaNOMA}, completing the proof. 
\end{proof}
The solution of each \eqref{eq:optimalZetaNOMA} can easily be performed using Newton's or bisection method in $O(N log(\epsilon))$, where $\epsilon$ is the required accuracy. Once $\boldsymbol\zeta^*$ has been obtained, $\boldsymbol\beta^*$ is computed from \eqref{eq:optimalBetaNOMA} in $O(N)$, yielding a total computational complexity of $O(N^2log(\epsilon))$ or $O(Nlog(\epsilon))$ if solved in parallel. Note that the methodology proposed to optimally solve \eqref{eq:outageNOMA2}, can be applied to any underlying outage probability distribution.

\section{Numerical Results}

We evaluate the performance of the developed algorithms by solving $1000$ randomly generated problems for different parameter configurations. All problem instances follow the system model introduced in Section \ref{sec:model} with $\sigma_h^2=1$ and $\sigma_n^2=1$. The fairness performance of the conventional TDMA scheme is used as a benchmark since it refers to the orthogonal allocation of the available degrees of freedom and is equivalent to any orthogonal multiple access scheme (\cite{tse}, Sec. 6.1.3); for the instantaneous CSI case, the optimal TDMA allocation can be solved by using a methodology similar to Algorithm \ref{alg:bisect}, while the optimal outage probability is obtained for equal time-sharing and power-split.

Fig. 1 demonstrates the achievable maximum fairness rate for different $P$ and $N$. As expected, increasing $P$ or reducing $N$ improves the achievable fairness rate. Interestingly, the performance gain from $N=10$ to $N=5$ is significantly higher than from $N=20$ to $N=10$. Notice also that as $P$ increases the rate of improvement reduces because the fairness data rate is a logarithmic function of power. 

Fig. 2 provides a comparison between NOMA and TDMA by depicting the achievable fairness rate for different $P$ for varying number of users. Clearly, NOMA is significantly better than TDMA. In fact, as the number of users increases, the advantage of NOMA over TDMA increases almost linearly. NOMA is also substantially better than TDMA in terms of computational complexity, as TDMA requires the solution of a sequence of convex programs. 

Figs. 3 compares the outage probability between NOMA and TDMA under average CSI information as a function of the total transmitted power for target spectral efficiency $r_0=0.05$ BPCU and $r_0=0.50$ BPCU, for $N=5$. In TDMA the optimal outage probability is obtained for equal time-split ratio and power among users. Results are also shown for a fixed NOMA PA scheme proposed in \cite{DIN}, with $\beta_m = \frac{N - m + 1}{\mu}$, $m\in\mathcal{N}$, where $\mu$ is selected such that $\sum_{m\in\mathcal{N}} \beta_m = 1$. It can be easily observed that in all cases considered, NOMA outperforms TDMA by an order of magnitude. NOMA also has at least five times better performance compared to the fixed NOMA PA scheme. As expected, higher target spectral efficiency results in worse outage probability because it is more difficult to be satisfied. For validation purposes, outage probabilities were obtained numerically for $10^8$ problems instances illustrating maximum absolute relative error 0.7\% for all cases considered.

\vspace{-0.2cm}
\section{Conclusions}
In this paper, the problem of optimal PA to maximize fairness among users of a NOMA downlink system is investigated in terms of data-rate under full CSI and outage probability under average CSI. Although the resulting problems are non-convex, simple low-complexity algorithms are developed that provide the optimal solution. Simulation results demonstrate the efficiency of NOMA, achieving fairness performance that is approximately an order of magnitude better than TDMA in the considered configurations. The main results of this work show that NOMA can ensure high fairness requirements through appropriate PA and is a promising MA scheme for future 5G communication systems.

\end{document}